\documentclass[11pt,a4paper]{article}

\usepackage[utf8]{inputenc}
\usepackage[russian, english ]{babel}

\usepackage[table]{xcolor}
\usepackage{authblk}
\usepackage[margin=1.2in]{geometry}
\usepackage{graphicx}
\usepackage{subcaption}
\usepackage{amssymb,amsmath,amsthm}
\usepackage{delarray}
\usepackage{enumitem}
\usepackage{fourier}
\usepackage{tikz,calc}
\usetikzlibrary{calc}
\usepackage{hyperref}

\graphicspath{{figures/}}

\newtheorem{theorem}{\textbf{Theorem}}
\newtheorem{lemma}{\textbf{Lemma}}
\newtheorem{proposition}{\textbf{Proposition}}
\newtheorem{corollary}{\textbf{Corollary}}
\newtheorem{remark}{\textbf{Remark}}
\newtheorem{remark*}{\textbf{Remark}}

\newtheorem{definition}{\textbf{Definition}}

\newtheorem{example}{\textbf{Example}}

\newcommand{\N}{\mathbb{N}}
\newcommand{\Z}{\mathbb{Z}}
\newcommand{\R}{\mathbb{R}}
\newcommand{\CC}{\mathbb{C}}
\newcommand{\real}{\mathrm{Re}}

\renewcommand{\int}[1]{[{#1}]}

\newcommand{\distance}{\mathbf{d}}

\newcommand{\polygon}{\chi} % the characteristic polygon of a multigrid dual tiling
\newcommand{\dominantlines}{\mathbf{L}}
\newcommand{\hausdorff}{\mathrm{d}_H} % the hausdorff distance
\newcommand{\cone}{\angle} %{\mathrm{cone}} % a cone
\newcommand{\tendsto}[1]{\underset{{#1}\to\infty}{\longrightarrow}} % tends to something as #1 tends to infinity

\newcommand{\overbow}[1]{
   \tikz [baseline = (N.base), every node/.style={}] {
      \node [inner sep = 0pt] (N) {$#1$};
      \draw [line width = 0.4pt] plot [smooth, tension=1.3] coordinates {
         ($(N.north west) + (0.1ex,0)$)
         ($(N.north)      + (0,0.5ex)$)
         ($(N.north east) + (0,0)$)
      };
   }
} %% \overbow is a bow over letters in mathmode to symbolyse the union of tiles, or the induced polygon

\newcommand{\grid}{G}
\newcommand{\multigrid}{H}
\newcommand{\tiling}{T}
\newcommand{\dual}[1]{\widetilde{#1}}
\newcommand{\dualization}{F}
\newcommand{\lindual}{\mathcal{F}}
\newcommand{\imag}{\mathrm{i}}
\newcommand{\diam}{\mathrm{diam}}
\newcommand{\dominantlineoffset}{\nu}

\title{Polygonal corona limit on multigrid dual tilings}
\author[1,3,4]{Victor Lutfalla}
\author[1,2]{K\'evin Perrot}
\affil[1]{Universit\'e publique}%, France}
\affil[2]{Aix-Marseille Univ., Univ. Toulon, CNRS, LIS, UMR 7020, Marseille, France}
\affil[3]{Univ. de Caen, CNRS, GREYC, UMR 6072, Caen, France}
\affil[4]{Aix-Marseille Univ., CNRS, I2M, UMR 7373, Marseille, France}
\date{}

\begin{document}
\renewcommand{\labelitemi}{$\circ$}
\renewcommand{\labelitemii}{$\circ$}
\setlist[itemize,enumerate]{nosep}
\maketitle

%%%%%%%%%%%%%%%%%%%%%%%%%%%%%%%%
\begin{abstract}
  The growth pattern of an invasive cell-to-cell propagation
  (called the successive coronas) on the square grid is a tilted square.
  On the triangular and hexagonal grids, it is an hexagon.
  It is remarkable that, on the aperiodic structure of Penrose tilings,
  this cell-to-cell diffusion process tends to a regular decagon (at the limit).
  On any multigrid dual tiling, it tends to a polygon which we call characteristic polygon.
  In this article we provide a complete and self-contained proof of this result.
  Exploiting this elegant duality allows to fully understand why such surprising phenomena,
  of seeing highly regular polygonal shapes emerge from aperiodic underlying structures, happen.
\end{abstract}

%%%%%%%%%%%%%%%%%%%%%%%%%%%%%%%%
\section{Introduction}

A geometric edge-to-edge tiling is a covering of the two dimensional plane by polygonal tiles,
such that there is no hole, no two tiles overlap, and adjacent tiles share a full edge.
When all tiles are rhombuses (have four sides of equal lengths),
we call it an edge-to-edge rhombus tiling.
The most famous examples are certainly Penrose tilings~\cite{p74}.
They received great attention, for their aesthetics and combinatorial properties
in connexion with the growth of quasicrystals~\cite{bg17}.
They have the property of being aperiodic (no translation vector leaves the tiling unchanged),
and quasiperiodic (every finite subset of tiles appears infinitely often).

In~\cite{ai16}, it has been proven that, at first sight surprisingly, regular decagons
emerge as fundamental elements of the structure of Penrose tilings.
From an initial finite set of selected tiles called patch,
the edge-to-edge diffusion process
(at each discrete step, tiles adjacent to the current selected patch are included in the selection)
produces a regular decagon at the limit (after renormalization).
This limit shape is called the corona limit.
The authors of~\cite{ai16} studied the corona limits of Penrose tilings
through the pattern of signal propagation in a simple growth cellular automata,
using combinatorial tools related to local dynamical behaviors specific to Penrose tilings (Ammann bars).
Regular corona limits obviously appear on simple periodic tilings such as triangular, square and hexagonal grids,
and have also been characterized on other periodic tilings~\cite{acik19}.
The corona limits of multigrid dual tilings are also discussed in~\cite{dhf22},
where the authors state a similar result without a full formal proof.

In the 1980s, de Bruijn discovered a duality between a class of edge-to-edge rhombus tilings
and regular multigrids~\cite{db81,db86}.
That is, the former are obtained from a finite family of infinitely many evenly spaced parallel lines
(along a finite set of directions called normal vectors;
each group of parallel lines also has a reference position called offset)
by a correspondence associating a tile to each intersection of two lines.
The corresponding tile is, in terms of Euclidean distance, not far from the intersection (up to a uniform linear map).
To lever the results of~\cite{ai16}, we first consider corona limits on the multigrid,
in order to take advantage of its embodiment of the tiling's regularities.
During a second step, we transfer our characterization of corona limits on the multigrid,
to the dual edge-to-edge rhombus tilings.

Limit shapes of growth processes on $\R^d$ (and $\Z^d$) have been studied in~\cite{gg93},
for threshold dynamics ($\theta>0$) defined over a given finite neighborhood ($N\subset\R^d$):
an element $x\in\R^d$ is added to the current selection $A\subset\R^d$ at the next step when
the Lebesgues measure of its set of selected neighbors reaches the threshold ($|A\cap(x+N)|\geq\theta$).
It has been proven that, except in degenerated cases,
there is a unique limit shape which is a polygon.
Growth processes on groups also received attention.
In~\cite{em23} the authors consider strongly connected groups $\Gamma$
and the iterated Minkowski sum $\phi_A(W)=\{wa \mid w\in W, a\in A\}$.
When $A$ contains the neutral element and generates $\Gamma$,
it has been proven that $\phi_A^n(\{v\})$ tends to $\Gamma$ (as $n$ tends to $\infty$).
Moreover, the growth of iterated Minkowski sums $\phi_A^n(\{v\})$
provide fine informations related to the geometry of the group.

The edge-to-edge propagation of coronas on tilings can be described in terms of sandpiles~\cite{btw87}.
A sandpile configuration associates an integer to each tile, which is its number of sand grains.
When a tile has as many grains as adjacent tiles, it gives one grain to each of its adjacent tiles.
Considering the maximum stable configuration having, on each tile,
one less grain than its number of adjacent tiles ($3$ in the case of rhombuses),
a single grain addition at tile $t$ triggers a chain reaction of topplings that corresponds
to computing the successive coronas, starting from the initial patch $\{t\}$.
In this regards, reference \cite{fp22} experiments the roundness of shapes obtained
when adding the so called identity element of the sandpile group (on finite tilings).
One may play with coronas using JS-Sandpiles,
a sandpile simulator written in javascript
(ready to play in your browser).\\
\centerline{\url{https://github.com/huacayacauh/JS-Sandpile/}}
Simply generate a tiling, set the maximum stable configuration, add one (or multiple) sand grains, then play.
Figure~\ref{fig:corona_example_big} has been exported from this software.

We define multigrids and their dual tilings in Section~\ref{s:multigrid},
where we state our main result: the corona limit of any multigrid dual tiling 
is a polygon (defined as its characteristic polygon)
having parallel opposite sides whose directions are given by the multigrid normal vectors.
Section~\ref{s:corona} splits the proof into three steps:
studying first the corona along individual multigrid lines,
second in the full multigrid,
third transfer the result to the dual tiling.
In Section~\ref{s:conclusion} we conclude and present perspectives.

As mentionned above, the main result that we present here might not be considered as new
because a similar one was stated in~\cite{dhf22} in the context of crystallography.
However, the authors do not provide a full formal proof but rely heavily on
crystallographic \emph{folklore} results and informally reuse previous statements
that were initially stated and proved for the periodic case~\cite{z01, z02}. %% NOTE removed Maalev-Shutov 2011 as it makes little sense
Here we provide a complete, self-contained, and independent proof of this result.

%%%%%%%%%%%%%%%%%%%%%%%%%%%%%%%%
\section{Multigrids, dual tilings, and coronas}
\label{s:multigrid}

We identify $\R^2$ with $\CC$,
denote $\overline{\zeta}$ the complex conjugate of $\zeta\in\CC$,
and $\zeta^\bot$ its orthogonal.
We denote the imaginary number $\imag$, it is usually easily distinguishable from context with the integer index $i$.
Given $z,z'\in\CC$, let $z\cdot z'$ denote their scalar product,
\emph{i.e.}, $z\cdot z':=\real(z\overline{z'})$. 
We present multigrids and their dual tilings, which are rhombus tilings (Lemma~\ref{lemma:rhombus}).
Then we introduce the corona limit of a patch, and state our main result (Theorem~\ref{theorem:corona}).

\begin{definition}[Grid and multigrid]
  We call \emph{grid} of normal vector $\zeta\in\CC$ (such that $|\zeta|=1$) and offset $\gamma\in\R$ the set of parallel lines
  \[
    \grid(\zeta,\gamma) := \{ z \in \mathbb{C}, z\cdot \zeta - \gamma \in \mathbb{Z}\}.
  \]

  We call \emph{multigrid} a union of grids.
  Given an integer $d\in\N_+$, a $d$-tuple $\zeta = (\zeta_i)_{0\leq i < d}$ and a $d$-tuple $\gamma = (\gamma_i)_{0\leq i <d}$,
  the multigrid of normal vectors $\zeta$ and offsets $\gamma$ is defined as
  \[
    \multigrid(\zeta,\gamma) := \bigcup_{0\leq i <d} \grid(\zeta_i, \gamma_i).
  \]

  A multigrid is called \emph{singular} when there exists a point where at least 3 lines (of different grids) intersect, and \emph{regular} otherwise.
  
%% A \emph{multigrid} is a collection of $d\in\N_+$ directions and shifts,
%%    $(\zeta_i,\theta_i)_{i\in\int{d}}\in(\CC\times\CC)^d$.
\end{definition}

\begin{figure}
  \center
  \begin{subfigure}{0.45\textwidth}
    \includegraphics[width=\textwidth]{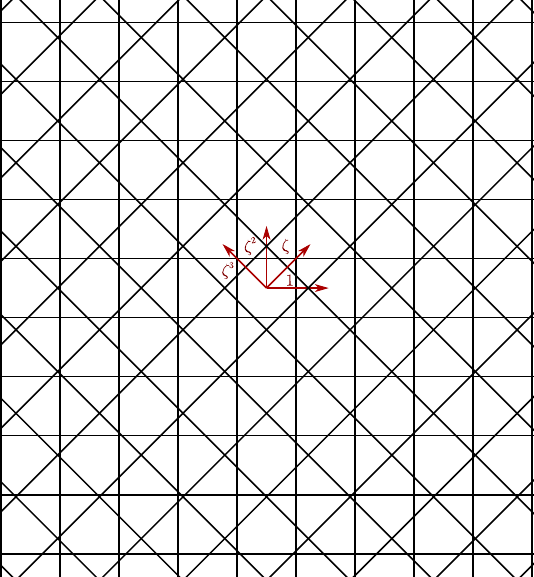}
    \caption{
      Multigrid of normal vectors $\zeta = (e^{i\pi/4})_{0\leq i < 4}$
      and offsets $\tfrac{1}{2}$~\cite{b82}.
    }
  \end{subfigure}
  \hspace*{1cm}
  \begin{subfigure}{0.45\textwidth}
    \includegraphics[width=\textwidth]{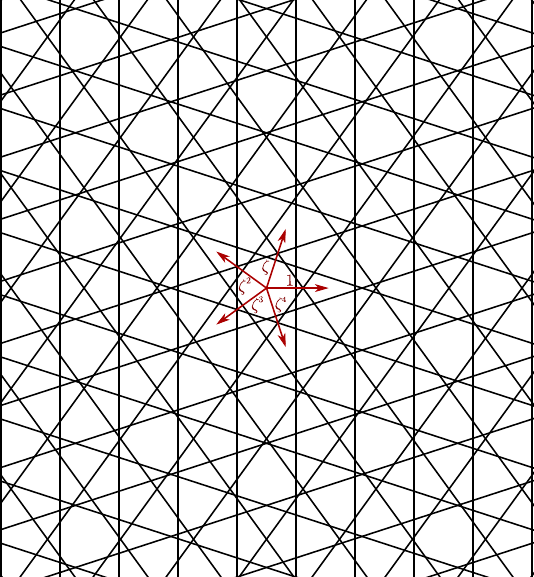}
    \caption{
      Pentagrid of offsets $\tfrac{1}{2}$~\cite{db81}.
    }
  \end{subfigure}
  \caption{Examples of multigrids.}
  \label{fig:multigrid_example}
\end{figure}

When $d$ is odd and $\zeta = (e^{2i\pi/d})_{0\leq i < d}$ we call $\multigrid(\zeta,\gamma)$ the $d$-fold multigrid of offsets $\gamma$.
When $d=5$ we call it the pentagrid of offsets $\gamma$~\cite{db81}.
See Figure~\ref{fig:multigrid_example}.

\begin{definition}[Dualization and multigrid dual tiling]
  Given a multigrid $\multigrid(\zeta,\gamma)$, we define the \emph{dualization function} $\dualization:\CC\to\CC$ as
  \[
    \dualization(z) := \sum\limits_{0\leq i<d} \left\lceil z\cdot \zeta_i- \gamma_i\right\rceil\zeta_i.
  \]

  The \emph{multigrid dual tiling} $\tiling(\zeta,\gamma)$ is defined by its set of vertices $V$ and of edges $E$ as
  \[
    V:=\dualization(\CC) \qquad
    E:= \{ \{z,z'\} \mid z,z' \in V \text{ and } \exists i, z' = z + \zeta_i\}.
    \]
    \label{def:dualization}
\end{definition}

\begin{figure}
  \begin{subfigure}{0.45\textwidth}
    \includegraphics[width=\textwidth]{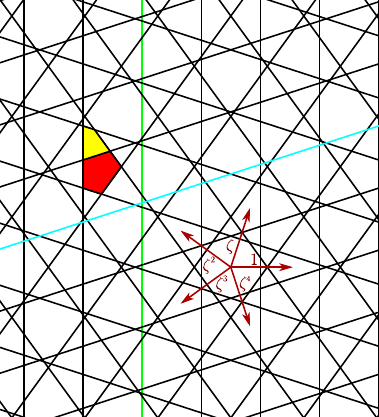}
    \caption{Pentagrid of offsets $\tfrac{1}{2}$.}
  \end{subfigure}
  \hspace*{1cm}
  \begin{subfigure}{0.45\textwidth}
    \includegraphics[width=\textwidth]{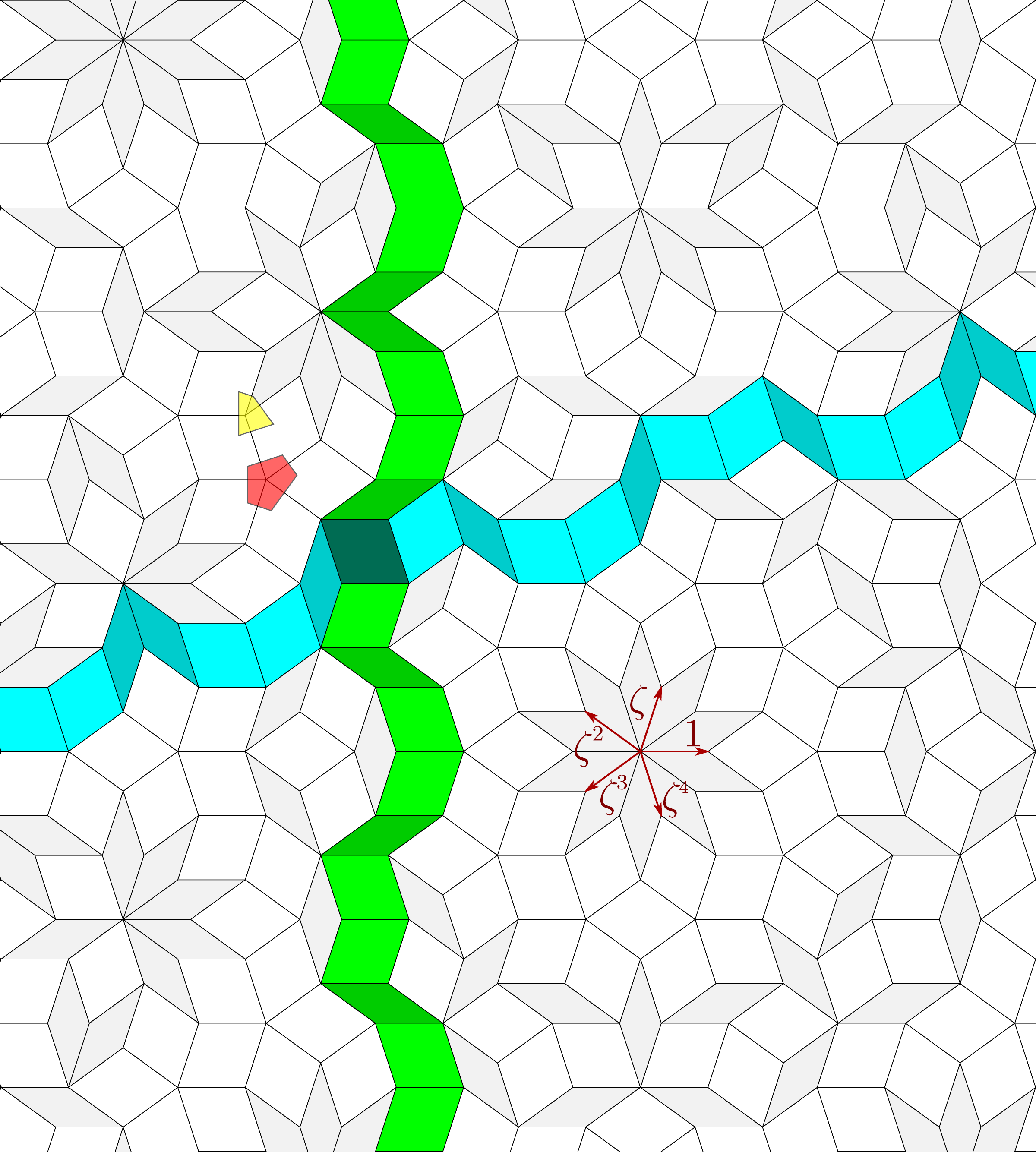}
    \caption{Dual rhombus tiling.}
    \end{subfigure}
  \caption{
    Example of a regular pentagrid and its dual tiling.
    Some elements of the multigrid and their dual in the tiling have been colored:
    in red and yellow two cells of the multigrid and the corresponding vertices in the tiling,
    and in blue and green two lines of the multigrid and the corresponding ribbons of tiles
    (each intersection of two grid lines corresponds to one tile).
  }
  \label{fig:multigrid_and_dual}
\end{figure}

Observe that $\dualization(\CC)$ is a countable set because of the ceiling operation in the defintion of $\dualization$.

\begin{lemma}[Regular multigrid dual tiling~\cite{db81,db86}]
  The dual tiling of a regular multigrid is an edge-to-edge rhombus tiling.
  \label{lemma:rhombus}
\end{lemma}

In the following, we will mainly employ the multigrid setting to study the corona limit,
and eventually go to the dual tiling at the end.
Each intersection (crossing) of two multigrid lines corresponds to a tile.
See Figure~\ref{fig:multigrid_and_dual}.
Penrose rhombus tilings~\cite{p74,p79} are dual of regular multigrids (see Remark~\ref{remark:penrose}).
%Penrose (required for Figure~\ref{fig:corona_def}, but defined only in Remark~\ref{remark:penrose}),
%multgrids~\cite{db81, db86}.

\begin{definition}[Types of crossings and tiles]
  We call \emph{$i$-line} or \emph{line of type $i$} a line of the grid $\grid(\zeta_i,\gamma_i)$.
  We call \emph{crossing of type $(i,j)$} the intersection point of a line of type $i$ with a line of type $j$.
  We call \emph{tile of type $(i,j)$} a tile that is the dual of a crossing of type $(i,j)$.
\end{definition}

Note that by definition a tile of type $(i,j)$ has edges along $\zeta_i$ and $\zeta_j$.
Given a tiling $\tiling$, denote $t\sim_\tiling t'$ when tiles $t,t'\in \tiling$ share an edge,
\emph{i.e.} they are adjacent in terms of the graph $\multigrid(\zeta,\gamma)$
(where each crossing is a vertex corresponding to a tile, and adjacencies are given by multigrid segments).
When the tiling is clear from the context,
we drop the subscript notation and simply denote $\sim$.
A \emph{patch} $P$ is a finite connected subset of tiles.
We denote $t\sim P$ when tile $t$ is adjacent to a tile of patch $P$ and $t\notin P$.

%% %% figures of sandpile identities, where to put them?
%% \includegraphics[width=3cm]{pics-P3-cutandproject-1-id.pdf}
%% \includegraphics[width=3cm]{pics-P3-cutandproject-2-id.pdf}
%% \includegraphics[width=3cm]{pics-P3-cutandproject-3-id.pdf}
%% \includegraphics[width=3cm]{pics-P3-cutandproject-4-id.pdf}

Given a patch $P$, we denote $\overbow{P}$ the polygon associated to patch $P$,
\emph{i.e.} the contour of $P$ as a collection of points linked by segments.
For a point $x\in\R^2$, we denote $x\in P$ when $x$ lies within a tile of $P$.
Given a two dimensional polygon $\Delta$, a scalar $\lambda\in\R^{+}$ and a point $x\in\R^2$,
we denote $\frac{\Delta_x}{\lambda}$ the polygon obtained by homothety of center $x$ and ratio $\lambda$.
We drop $x$ from the notation when $x$ is the origin.

\begin{definition}[Corona]
  Given an edge-to-edge tiling $\tiling$ and a finite patch $P$,
  we call the $n$-th corona of $P$ in $\tiling$ the patch of tiles that are at edge-distance at most $n$ from $P$.
  Formally, we define the sequence $(P_n)_{n\in\N}$ with $P_0=P$
  and $P_{n+1}=P_n\cup\{t\mid t\sim P_n\}$.
\end{definition}

Examples of corona are depicted on Figure~\ref{fig:corona_def}.
The corona limit corresponds to the normalized limit of the iterated coronas.

\begin{figure}
  \center
  \includegraphics[width=0.5\textwidth]{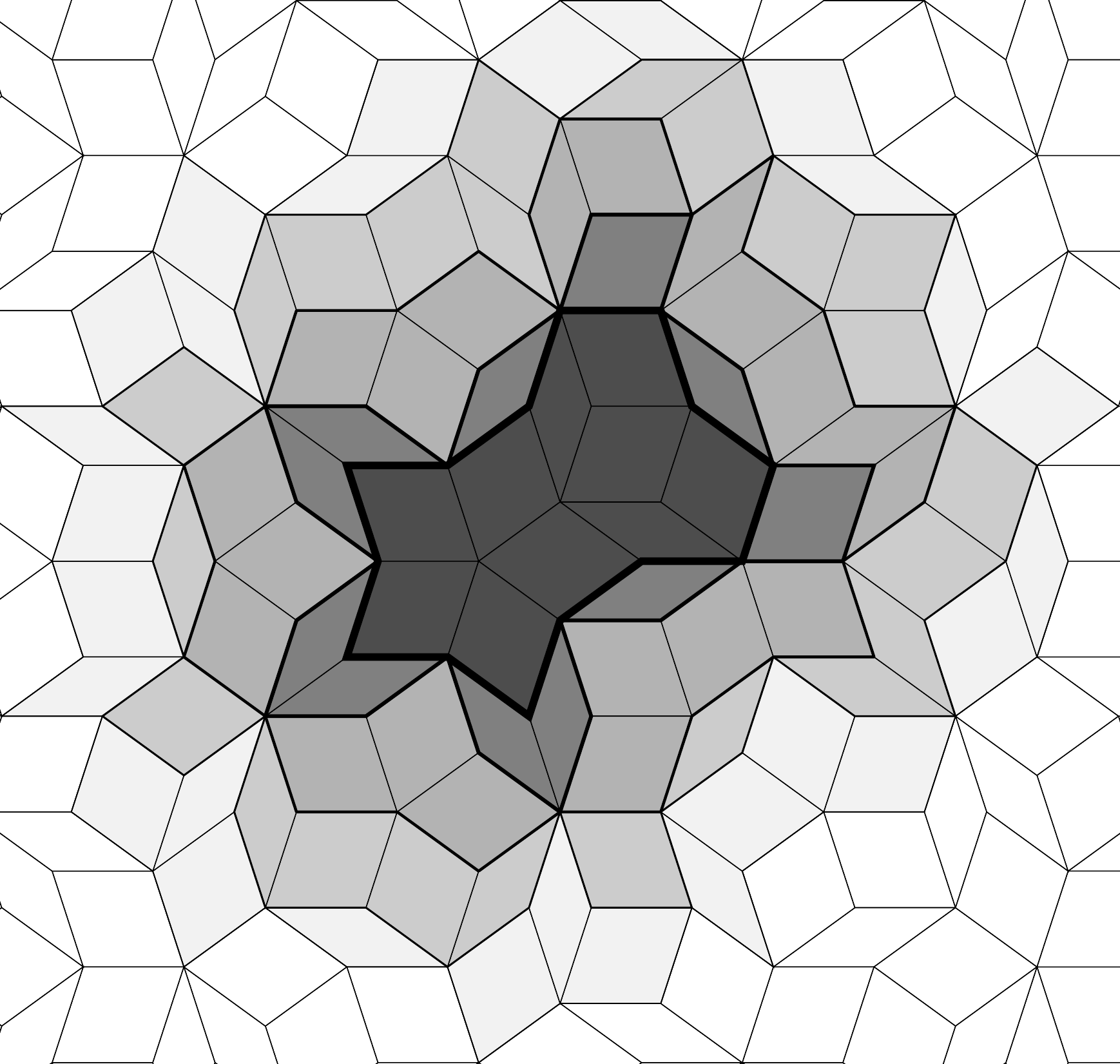}
  \caption{A patch of a Penrose tiling in dark grey and its 4 first coronas in greyscale.}
  \label{fig:corona_def}
\end{figure}

\begin{definition}[Corona limit]
  Polygon $\Delta$ is the \emph{corona limit} of patch $P$ in tiling $\tiling$
  when 
  \[
    \frac{\overbow{P_n}}{n} \tendsto{n} \Delta
  \]
  where $(P_n)_{n\in\N}$ is the corona sequence of $P$ and the convergence is for the Hausdorff metric.
  Polygon $\Delta$ is the \emph{uniform corona limit} of $\tiling$ when for any patch $P\subset \tiling$, $\Delta$ is the corona limit of $P$.
  \label{def:coronalimit}
\end{definition}

See Figure~\ref{fig:corona_example_big} for an illustration.

\begin{figure}
  \center
  \begin{subfigure}{0.3\textwidth}
    \includegraphics[width=\textwidth]{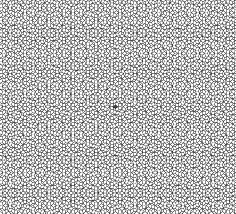}
    \caption{A single initial tile $P_0$.}
  \end{subfigure}
  ~
  \begin{subfigure}{0.3\textwidth}
    \includegraphics[width=\textwidth]{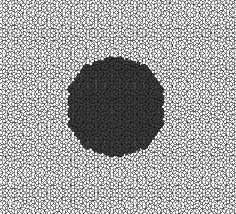}
    \caption{The 20th corona $P_{20}$.}
  \end{subfigure}
  ~
  \begin{subfigure}{0.3\textwidth}
    \includegraphics[width=\textwidth]{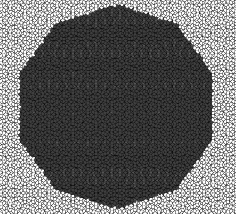}
    \caption{The 40-th corona $P_{40}$.}
  \end{subfigure}
  \caption{Successive coronas from an initial single tile in a Penrose tiling.
    The 40-th corona $P_{40}$ is already very close to a regular decagon which is the corona limit of Penrose tilings.}
  \label{fig:corona_example_big}
\end{figure}

%% Kevin: j'ai mis cette figure là dans le source
%%        pour que LaTeX la place avec les deux précédentes sur une seule page.
\begin{figure}
  \center
  \includegraphics[width=0.6\textwidth]{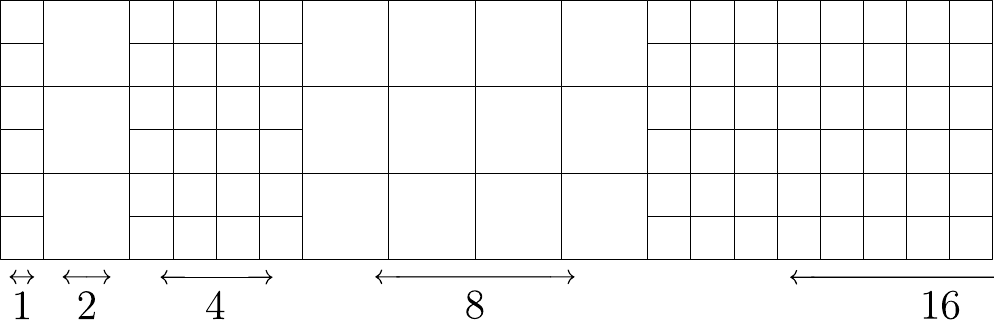}
  \caption{A simple tiling by $1\times1$ and $2\times 2$ square tiles where the corona limit does not exist.}
  \label{fig:no-corona}
\end{figure}

%\begin{remark}
%  A corona is not always simply connected (\emph{e.g.}~when the initial patch is not simply connected),
%  but it always eventually becomes simply connected
%  (provided the tiling is simply connected).
%  %Moreover, the convergence to $\Delta$ is independent of the choice of point $x$ in Definition~\ref{def:coronalimit}:
%  %if there exists an $x\in\R^2$ such that $\Delta$ is the corona limit of patch $P$
%  %then the convergence holds for any $x\in\R^2$.
%  \label{remark:patch_connected}
%\end{remark}

The uniformity condition (for any patch) can be restricted to single tiles,
as demonstrated in the following lemma.

\begin{lemma}[Corona limit of single tiles]
  Let $\Delta$ be a polygon and $\tiling$ a rhombus tiling.
  If for any single tile $t\in \tiling$, $\Delta$ is the corona limit of patch $\{t\}$,
  then $\Delta$ is the uniform corona limit of $\tiling$.
\end{lemma}

\begin{proof}
  Let $\tiling$ be an edge-to-edge rhombus tiling, and $\Delta$ a polygon such that for any single tile $t$ the corona limit of $\{t\}$ is $\Delta$.
  Let $P$ be a patch. Recall that a patch is a finite simply connected set of non overlapping tiles.
  Then $P = \bigcup_{i=0}^{m-1} \{t_i\}$ and for any $n$, $P_n = \bigcup_{i=0}^{m-1} \{t_i\}_n$.
  Since for any $i\in \{0,\dots,m-1\}$ we have $\frac{\overbow{\{t_i\}_n}}{n} \to \Delta$ we also have $\tfrac{\overbow{P_n}}{n} \to \Delta$.
\end{proof}

\begin{remark}
  The definition of corona limit behaves nicely on sufficiently uniform rhombus tilings such as the multigrid dual tilings,
  but may not exist in general.
  For example, consider a tiling by $1\times1$ and $2\times 2$ square tiles aranged in vertical strips of increasing width, see Figure \ref{fig:no-corona}.
  In such a tiling, the corona limit does not exist as the coronas spread with speed $1$ in the strips of $1\times 1$ squares and with speed $2$ in the strips of $2\times 2$ squares.
  So if both the $1$-strips and the $2$-strips have unbounded width there is no ``average growth speed'' and therefore no corona limit.
\end{remark}

In the case of rhombus tilings dual of regular multigrids,
corona limits always exist and the polygon can be computed from the normal vectors defining the multigrid.
We call it the characteristic polygon of the multigrid/tiling,
and state our main result (Theorem~\ref{theorem:corona}).

\begin{definition}[Characteristic polygon $\polygon$]
  Given a regular multigrid $\multigrid(\zeta,\gamma)$ and its dual tiling $\tiling(\zeta,\gamma)$.
  Let $\polygon$ be the $2d$-gone with vertices $\pm \polygon_i \zeta_i^\bot$ for $0\leq i < d$ with
  \[
    \polygon_i := \left(\sum\limits_{0\leq j < d}|\zeta_i^\bot \cdot \zeta_j| \right)^{-1}.
    \]

  Let $\dual{\polygon}$ be the $2d$-gone with vertices $\pm \dual{\polygon_i}$ for $0\leq i <d$ with
  \[
    \dual{\polygon_i}
    := \polygon_i\sum\limits_{0\leq j < d}(\zeta_i^\bot \cdot \zeta_j) \zeta_j
    = \frac{\sum\limits_{0\leq j < d}(\zeta_i^\bot \cdot\zeta_j) \zeta_j}{\sum\limits_{0\leq j < 0}|\zeta_i^\bot\cdot\zeta_j|}.
  \]

  We say that $\polygon$ and $\dual{\polygon}$ are the characteristic polygons
  of $\multigrid(\zeta,\gamma)$ and $\tiling(\zeta,\gamma)$, respectively.
  % multigrid characteristic polygon \polygon
  % multigrid dual tiling polygon \dual{\polygon}
  \label{def:polygon}
\end{definition}

The value $|\zeta_i^\bot\cdot\zeta_j|^{-1}$ is the distance along a $i$-line
between two intersection points with a $j$-line in the multigrid,
so $|\zeta_i^\bot\cdot\zeta_j|$ is the frequency of $(i,j)$ points on a $i$-line.
Hence $\polygon_i$ is the inverse of the sum of the frequencies of intersection points along a $i$-line,
\emph{i.e.}, $\polygon_i$ is the average distance between consecutive multigrid vertices along a $i$-line.
This explains why the growth of the corona limit along a $i$-line is $\polygon_i$.
In the definition of $\dual{\polygon_i}$, we can recognize a formula similar
to the dualization function $\dualization$ from Definition \ref{def:dualization},
up to the absence of floor $\lfloor~\rfloor$ and offsets $\gamma$.
Essentially, $\dual{\polygon_i} \approx \dualization(\polygon_i)$.
%%  QUESTION {Give geometrical intuitions the formulas for $\polygon_i$ and $\dual{\polygon_i}$.}

\begin{example}[Characteristic polygon for Penrose tilings]
  In the case of Penrose tilings the directions are the fifth roots of unit $\zeta_k = e^{\frac{2\imag k \pi}{5}}$ \cite{db81}.
  Therefore we have $\polygon_0 = \polygon_1 = \dots = \polygon_4 = \left(2\sin(\tfrac{2\pi}{5} + 2\sin(\tfrac{4\pi}{5})\right)^{-1} \approx 0.32$.
  This can be seen as the fact that in the pentagrid ($5$-fold multigrid),
  the average distance between consecutive intersection points along a line
  is the same in all directions, and is equal to $\polygon_0 \approx 0.32$. 
  From this we get that the vertices of $\polygon$ are $\pm \polygon_0 \zeta_k^\bot$ for $0\leq k < 5$, see Figure \ref{fig:charpolygon}.
  In particular, the characteristic polygon is a regular decagon.
  We obtain that
  $\dual{\polygon_0} = \polygon_0\sum_{0\leq j < d}(\zeta_0^\bot \cdot \zeta_j) \zeta_j = \polygon_0\imag(2\sin^2(\tfrac{2\pi}{5}) + 2\sin^2(\tfrac{4\pi}{5})) = \tfrac{5\polygon_0}{2}\imag \approx 0.81\imag$.
  Adapting this to all other directions, we get
  $\dual{\polygon_k} =\tfrac{5\polygon_0}{2}\zeta_k^\bot \approx 0.81\zeta_k^\bot$.
  Note that $\tfrac{5\polygon_0}{2}\approx 0.81$ is the expected width of rhombuses
  along a chain of rhombuses in a Penrose tiling,
  meaning the average width of fat and thin rhombuses weighted by their respective densities.

  \begin{figure}
    \centering
    \begin{subfigure}{0.45\textwidth}
      \includegraphics[width=\textwidth]{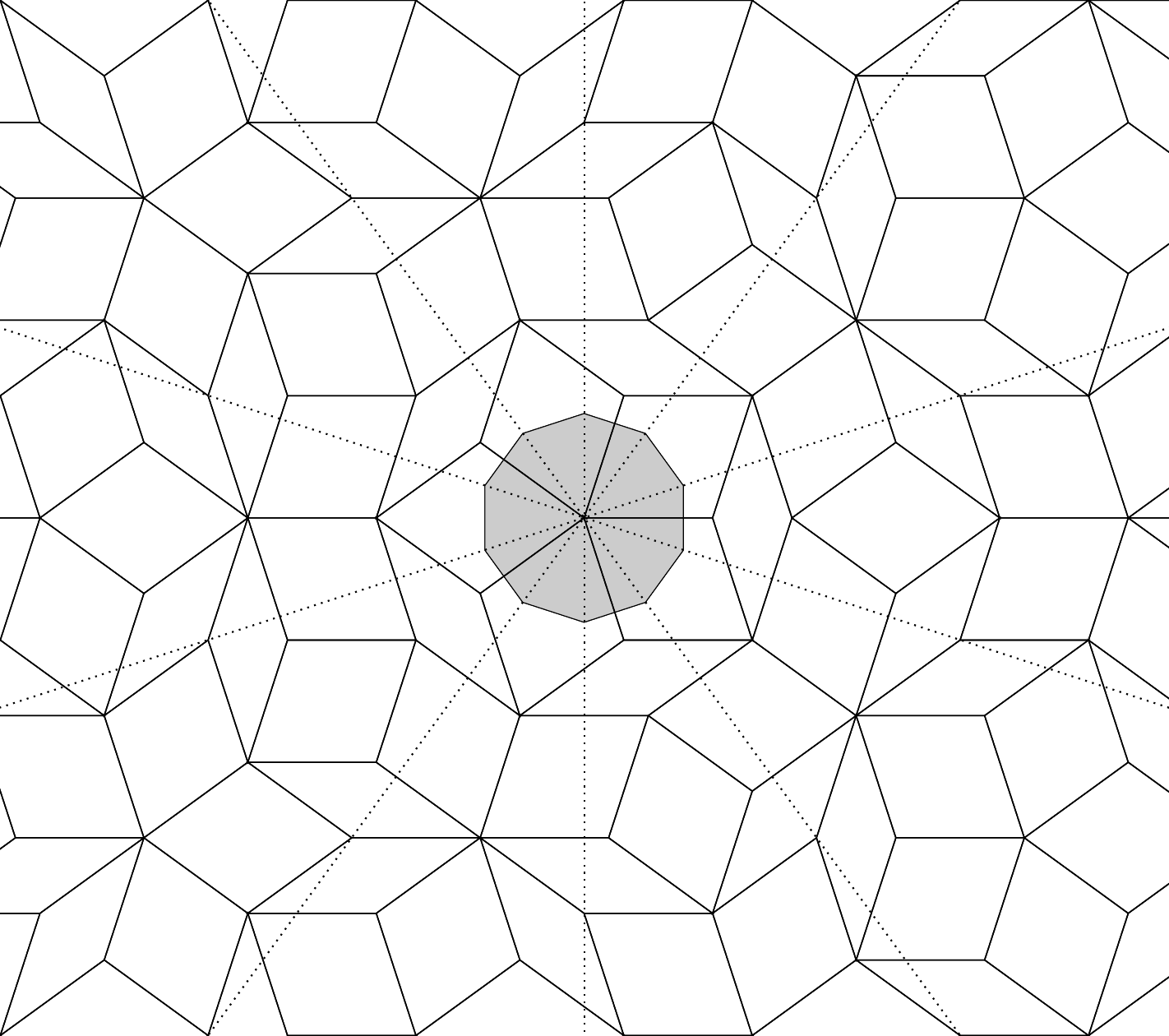}
      %\caption{Characteristic polygon of a Penrose tiling.}
    \end{subfigure}
    \hspace*{.5cm}
    \begin{subfigure}{0.45\textwidth}
      \center
      \includegraphics[width=0.9\textwidth]{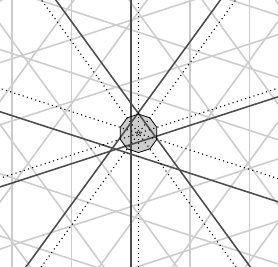}
      %\caption{Characteristic polygon in a pentagrid.}
      %\label{fig:G5_charpolygon}
    \end{subfigure}
    \caption{
      Characteristic polygons on Penrose tiling (left) and pentagrid (right).
      In dotted, the lines of orthogonal vector $\zeta_k$ passing through the origin.
    } %% NOTE les deux figures sont pas à la même échelle
    \label{fig:charpolygon}
  \end{figure}

\end{example}

Note that the characteristic $2d$-gones of multigrids and their dual tilings
have parallel opposite sides.
We can now state our main result and some observations.

\begin{theorem}[Corona limit]
  The corona limit of any multigrid dual tiling is its characteristic polygon.
  In particular:
  \begin{itemize}
  \item the corona limit only depends on the directions of the grid and not on its offsets,
  \item the corona limit of a $n$-fold multigrid dual tiling is a regular $2n$-gone.
  \end{itemize}
  \label{theorem:corona}
\end{theorem}

\begin{remark}[Pentagrids and Penrose tilings]
  Penrose rhombus tilings can be defined as the pentagrid dual tilings (or $5$-fold dual tilings)
  where the sum of the offsets is 0 modulo 1 \cite{db81,r96}.
  The corona limit of Penrose tilings is known to be a regular decagon \cite{ai16}.
  Our result proves that it is also the case for all pentagrid dual tilings even when the sum of the offsets is not 1.
  In particular, it also holds for the Antipenrose tilings where the sum of the offsets is $1/2$. 
  \label{remark:penrose}
\end{remark}

\begin{remark}[In cut-and-project terms]
  Multigrid dual tilings are cut-and-project tilings \cite{db86,gr86}.
  The result can be reformulated as the fact that the corona limit of a multigrid dual tiling
  depends only on its cut-and-project slope, and not on its cut-and-project intercept.
  For more insight on this result through the lens of cut-and-projection see \cite{dhf22},
  in particular the characteristic polygon $\dual{\chi}$ is in this context the projection
  of the intersection of the cut-and-project slope of the tiling with the orthoplex (the polyhedron
  with vertices $\pm e_i$ where $(e_i)_{0\leq i < d}$ is the canonical basis of the higher dimensional space).
\end{remark}

%%%%%%%%%%%%%%%%%%%%%%%%%%%%%%%%%
%\section{Sandpiles on undirected graphs}
%\label{s:sandpiles}
%
%original model~\cite{btw87},
%group~\cite{d90}.

%%%%%%%%%%%%%%%%%%%%%%%%%%%%%%%%
\section{Polygonal corona limit on multigrid dual tilings}
\label{s:corona}

For technical reasons we will first consider the corona limit on the multigrid itself
in Subsections~\ref{subsec:endpoints} (along individual lines of the multigrid)
and~\ref{subsec:corona-multigrid} (on the full multigrid),
before returning to the dual tiling in Subsection \ref{subsec:corona-tiling}
(thanks to the fact that up to a linear map the multigrid intersection points are close,
in terms of Euclidean distance, to their corresponding tile).
Indeed, we take advantage of the regularities within the multigrid, before transferring
the result to the tiling (which may have a non trivial aperiodic structure).
The characteristic polygon $\polygon$ (or $\dual{\polygon}$) embeds this reasoning:
its main directions of propagation are along individual lines (or ribbons),
and other multigrid intersection points (or other tiles)
adhere to this movement (this is the main technical difficulty).

We consider here the multigrid as a graph: the vertices are the crossing points of the grid lines
and two vertices are connected by an edge if they are two consecutive crossing points along a grid line.
Let us redefine some of the concepts in this setting.
We call \emph{patch} a finite and connected set of multigrid vertices,
and $n$-th \emph{corona} of a patch the set of multigrid vertices that are at distance at most $n$ from the patch.
To differenciate them from tiling patches, we denote $C$ a multigrid patch and
$(C_n)_{n\in\N}$ the corona sequence starting at $C_0=C$.

\begin{remark}
  Note that, even if $C_0$ is simply connected,
  it is not always the case that $C_i$ is simply-connected for all $i$.
  Consider for example a crescent shaped $C_0$.
  However, there always exists some $N$ such that for any
  $i\geq N$, the patch $C_i$ is simply-connected.
  \label{remark:patch_connected}
\end{remark}

For the rest of this section we fix an integer $d$,
a tuple of directions $\zeta = (\zeta_i)_{0\leq i < d}$
and a tuple of offsets $\gamma = (\gamma_i)_{0\leq i < d}$.
That is, we fix a regular multigrid $\multigrid(\zeta,\gamma)$
and dual tiling $\tiling(\zeta,\gamma)$.

%%%%%%%%%%%%%%%%%%%%%%%%%%%%%%%%
\subsection{Endpoints in the multigrid}
\label{subsec:endpoints}

We start the study of corona limits by considering their behavior along multigrid lines taken individually.
For this purpose, we select one line in each direction given by normal vectors $\zeta$.

\begin{definition}[Dominant line]
  In the multigrid $\multigrid$, we choose one line of each type,
  called \emph{dominant line}, and denote them $\dominantlines=(L_i)_{0\leq i <d}$.
\end{definition}

The intuition is that we take for each direction the multigrid line closest to the origin, and we study the corona sequence of a patch $C$ close to the origin, the typical case being a patch that intersects all the dominant lines.
%% \XXX{update}
%% Consider the case of a multigrid patch $C$ that crosses exactly one line of each direction
%% (see Figure \ref{fig:endpoints}), in this case simply take $L_i$ as the line of type $i$ that crosses $C$.
%% If more than one line of some type intersect $C$, we choose a dominant line for each direction.
%% On the other hand, if no line of some type intersect $C$, we choose any dominant line at finite distance.
%% Now the endpoints and their hulls concentrate on the growth of the corona limit along dominant lines.
Remark that for any initial multigrid patch $C_0$, there exist a $k$ such that the $k$-th corona $C_k$ intersects all dominant lines. This holds because each dominant line is at finite graph distance to $C_0$.
So without loss of generality, we assume in the rest of this section that $C_0$ intersects all dominant lines.

\begin{definition}[Endpoints $E_n$ and polygonal hull $\overbow{E_n}$.]
  Given an initial patch $C$,
  define the endpoints $E_n = (e_{n,i}^+, e_{n,i}^-)_{0\leq i < d}$,
  with $e_{n,i}^+$ and $e_{n,i}^-$ the intersection points at graph distance $n$
  from the initial patch $C$ along the dominant line $L_i$ in both directions.
  Define the polygonal hull $\overbow{E_n}$ as the polygon with vertices $\{e_{n,i}^+, e_{n,i}^-\}$.
\end{definition}

\begin{remark}
  %For small $n$, some $e_{n,i}$ and $e_{n,i'}$ might coincide, however 
  For large enough $n$ the $e_{n,i}$ are distinct and the polygonal hull $\overbow{E_n}$ is a $2d$-gone.
\end{remark}

Note that, since the multigrid has been chosen to be regular (no three lines intersect),
the $n$-th vertex along a line from a starting point is always well defined.
However, it might be hard to visualize as some intersection points might be arbitrarily close.

\begin{figure}
  \center
  \includegraphics[width=0.5\textwidth]{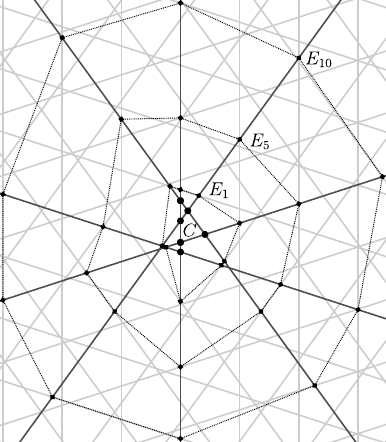}
  \caption{
    A pentagrid, a patch $C$ in round dots, and a selection of dominant lines 
    in darker grey.
    Endpoints $E_1$, $E_5$, $E_{10}$ for a direction are depicted,
    with the corresponding polygonal hull in dotted lines.
  }
  \label{fig:endpoints}
\end{figure}

The following lemma is illustrated on Figure~\ref{fig:frequency_crossing}.
It introduces a notation to consider the number of intersections between
one line in direction $\zeta_i$ and the grid (collection of parallel lines) in direction $\zeta_j$.

\begin{lemma}[Intersection points along a grid line]
  %% \TODO{$a$ is indistinguishable from text a, maybe replace with another symbol (maybe $\alpha$} %% XXX attention chiant à bien remplacer, et penser dans ce cas à remplacer dans la figure 8 fig:frequency_crossing
  Given a line $\ell$ of normal direction $\zeta_i$, a point $z\in \ell$ and a positive real number $\alpha$,
  denote $c_j(z,z+\alpha\zeta_i^\bot)$ the number of crossings of type $(i,j)$
  along $\ell$ between $z$ and $z+\alpha\zeta_i^\bot$.
  It is almost $\alpha$ times the absolute scalar product $|\zeta_i^\bot\cdot\zeta_j|$.
  More precisely, we have 
  \[
    \left|c_j(z,z+\alpha\zeta_i^\bot) - \alpha|\zeta_i^\bot\cdot\zeta_j| \right| \leq 2
  \]
  and in particular
  \[
    \frac{c_j(z,z+\alpha\zeta_i^\bot)}{\alpha}
    \underset{\alpha\to+\infty}{\longrightarrow}
    |\zeta_i^\bot\cdot\zeta_j|.
  \]
  \label{lemma:frequency_crossing}
\end{lemma}

\begin{figure}
  \center
  \includegraphics[width=0.5\textwidth]{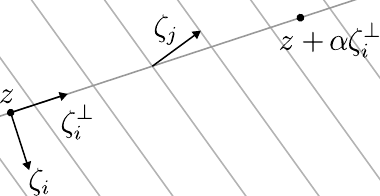}
  \caption{Crossings of type $(i,j)$ along a $i$-line.}
  \label{fig:frequency_crossing}
\end{figure}

%% Kevin: j'ai mis cette figure là dans le source
%%        pour que LaTeX la place avec les deux précédentes sur une seule page.
\begin{figure}
  \center
  \includegraphics[width=0.5\textwidth]{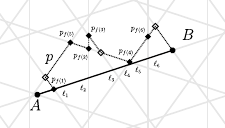}
  \caption{
    Two vertices $A$ and $B$ along a grid line.
    If the multigrid is regular then a shortest path between $A$ and $B$ is along the segment $AB$.
  }
  \label{fig:triangular_inequality_1}
\end{figure}

\begin{proof}
  The general idea is that the distance along a $i$-line between two consecutive $(i,j)$ crossings is precisely $(\zeta_i^\bot \cdot\zeta_j)^{-1}$.
  So along a $i$-line, on a segment of length $\alpha$, there are about $\alpha|\zeta_i^\bot \cdot \zeta_j|$ crossings of type $(i,j)$.

  Formally, an intersection of type $(i,j)$ is a complex $z$ such that $(z\cdot \zeta_j + \gamma_j) \in \mathbb{Z}$.
  So by definition, $c_j(z,z+\alpha\zeta_i^\bot)$ is precisely the number of complex numbers of the form $z + x\zeta_i^\bot$ with $x\in [0,\alpha]$ and such that $((z+x\zeta_j^\bot)\cdot \zeta_j + \gamma_j) \in \mathbb{Z}$.
  
  So we can reformulate it as 
  \[
    c_j(z,z+\alpha\zeta_i^\bot) =
    \left|\lfloor z\cdot\zeta_j+\gamma_j\rfloor - \lfloor(z+\alpha\zeta_i^\bot)\cdot\zeta_j + \gamma_j\rfloor \right|.
  \]
  From this we get the expected result with $\left|\lfloor x \rfloor - x \right| \leq 1$.
\end{proof}

The next corollaries are immediate consequences.

\begin{corollary}[Intersection points with all grids along a grid line]
  Given a line $\ell$ of normal direction $\zeta_i$, a point $z\in \ell$ and a positive real number $\alpha$,
  denote $c(z,z+\alpha\zeta_i^\bot)$ the tuple of the number of crossings of type $(i,j)$
  along $\ell$ between $z$ and $z+\alpha\zeta_i^\bot$ for all $j\neq i$, \emph{i.e.},
  \[
    c(z,z+\alpha\zeta_i^\bot) := \left( c_j(z,z+\alpha\zeta_i^\bot) \right)_{j\neq i}.
  \]
  Then $c(z,z+\alpha\zeta_i^\bot)$ is almost $\alpha$ times the tuple of (absolute) scalar products
  $|\zeta_i^\bot\cdot\zeta_j|$.
  More precisely, we have 
  \[
    \left\|c(z,z+\alpha\zeta_i^\bot) - \alpha\left(|\zeta_i^\bot\cdot\zeta_j|\right)_{j\neq i} \right\|_\infty \leq 2
  \]
  and in particular
  \[
    \frac{c(z,z+\alpha\zeta_i^\bot)}{\alpha} \underset{\alpha\to+\infty}{\longrightarrow} \left(|\zeta_i^\bot\cdot\zeta_j|\right)_{j\neq i}.
  \]
  \label{corollary:nth_vertex}
\end{corollary}

Recall that $\polygon_i=\left(\sum\limits_{0\leq j <d, j\neq i} |\zeta_i^\bot\cdot\zeta_j|\right)^{-1}$.

If we denote $n_c(\alpha)$ the sum of the $c_j(z, z+\alpha\zeta_i^\bot)$ for $j\neq i$,
we have $\frac{n_c(\alpha)}{\alpha}  \underset{\alpha\to+\infty}{\longrightarrow} \sum\limits_{j\neq i} |\zeta_i^\bot\cdot\zeta_j| = \polygon_i^{-1}$.
Note that since we inverse the role of $\alpha$ and $n$, we go from $n \approx \alpha \polygon_i^{-1}$ to $\alpha \approx n\polygon_i$. 
This is formalized in the next corollary.

\begin{corollary}$n$-th intersection point with a grid line]
  Given a line $\ell$ of normal direction $\zeta_i$ and a point $z\in \ell$,
  denote $z_n= z + \alpha_n\zeta_i^\bot$ the $n$-th intersection point along $\ell$ in direction $\zeta_i^\bot$.
  There exists a constant $\delta_0$ (depending only on the multigrid order and directions) such that for any $n\in\N$,
  \[
    \left| \alpha_n - n\polygon_i \right| \leq \delta_0
  \]
  and in particular
  \[
    \frac{\alpha_n}{n} \underset{n\to+\infty}{\longrightarrow} \polygon_i.
  \]
  \label{cor:nth_intersection}
\end{corollary}

We can now characterize the limit shape of $\overbow{E_n}$
in terms of the characteristic polygon of the multigrid $\multigrid(\zeta,\gamma)$.

\begin{proposition}[The limit shape of $\overbow{E_n}$ is $\polygon$]
  Let $C$ be a finite initial patch.
  There exists $\delta_1\in\mathbb{R}^+$ such that for any $n\in \N$,
  \[
    (n-\delta_1)\polygon \subseteq \overbow{E_n} \subseteq (n+\delta_1)\polygon
  \]
  and in particular, the limit shape of $\overbow{E_n}$ is the characteristic polygon $\polygon$.
  \label{prop:endpoints_limit_shape}
\end{proposition}

By convention, we consider that when $(n-\delta_1)\leq 0$ we have $(n-\delta_1)\polygon = \emptyset$
so that we only need to consider the case $n>\delta_1$.

\begin{proof}
  Let $z_i^+$ and $z_i^-$ be the extremal points of $C$ along its dominant grid line of normal $\zeta_i$ for $0\leq i < d$.
  Take $n\geq 0$, by Corollary~\ref{cor:nth_intersection} we have $|(e_{n,i}^+-z_i^+) - n\chi_i| \leq \delta_0$.
  Now take $\delta_1$ as
  \[
    \delta_1 := \frac{\delta_0 + \max\limits_{z\in C}|z|}{\min\limits_{0\leq i <d} |\polygon_i|}.
  \]
  We have $|e_{n,i}^+ - n\chi_i\zeta_i^\bot| \leq \delta_0 + |z_i^+| \leq \delta_i\polygon_i$,
  and similarly $|e_{n,i}^- + n\chi_i\zeta_i^\bot| \leq \delta_i\polygon_i$.
  Overall we get
  $(n-\delta_1)\polygon \subseteq \overbow{E_n} \subseteq (n+\delta_1)\polygon$,
  and therefore $\overbow{E_n}/n$ converges in Hausdorff distance to $\polygon$.
\end{proof}

%%%%%%%%%%%%%%%%%%%%%%%%%%%%%%%%
\subsection{Corona in the multigrid}
\label{subsec:corona-multigrid}

We now aim at considering all multigrid lines.
For this purpose, we start by considering the distance between pairs of crossings in the multigrid,
\emph{i.e.}~the shortest path between the corresponding vertices when we see the multigrid as a graph.
It can also be seen as the distance between tiles in the edge-to-edge rhombus tiling,
and is directly correlated to the growth of the corona sequence.

\begin{lemma}[Regular multigrid distance 1]
  Let $A$ and $B$ be two vertices along a grid line $\ell$.
  The straight path from $A$ to $B$ along $\ell$ in the multigrid
  is a shortest path for the graph distance.
  \label{lemma:multigrid_distance_1}
\end{lemma}

\begin{proof}
  See Figure \ref{fig:triangular_inequality_1} for an illustration.
  Recall that $\multigrid$ is a regular multigrid: no three lines intersect in a single point.
  Let $\ell$ be a grid line in $\multigrid$ and let $A$ and $B$ be two vertices on $\ell$.
  Let $p$ be a graph path from $A$ to $B$.
  We prove that $p$ has at least as many intermediate vertices as the straight segment $\ell_{AB}$.
  Let $k$ be the number of intermediate vertices in $\ell_{AB}$ and $m$ be the number of intermediate vertices on $p$.
  We write $p = (A, p_1, p_2, \dots, p_{m}, B)$ where $p_i$ is the $i$-th intermediate vertex along $p$ starting from $A$.

  Let $(\ell_i)_{0< i \leq k}$ be the lines that intersect $\ell_{AB}$ strictly between $A$ and $B$
  ($\multigrid$ is regular, hence each vertex along $\ell_{AB}$ corresponds to exactly one crossing line).
  Since the $\ell_i$ are (straight) lines, each $\ell_i$ that crosses $\ell_{AB}$ also crosses $p$,
  and each crossing is a vertex. Consequently, for each $i$, $\ell_i \cap p$ contains at least one vertex.
  Let $f$ be the function from $\llbracket 1, k\rrbracket$ to $\llbracket 1, m \rrbracket$
  that to $i$ associates the smallest $j$ such that $p_j$ belongs to $\ell_i$, \emph{i.e.},
  \[
    f(i):= \min \left\{ j \in \llbracket 1, m\rrbracket \mid  p_j \in \ell_i \right\}.
  \]

  We prove that $f$ is injective, thus $m\geq k$ and the length of path $p$ is at least as long as $\ell_{AB}$,
  which proves the lemma.
  For the sake of contradiction, assume we have $i\neq i'$ such that $f(i) = f(i') = j$.
  It follows that $\ell_i$ and $\ell_{i'}$ both intersect on vertex $p_j$ of $p$.
  However, since $\multigrid$ is regular, $p_j$ is exactly the intersection of $\ell_i$ and $\ell_{i'}$.
  If $j=1$, this means that the edge $p_Ap_1$ is a segment of $\ell_i$ or $\ell_{i'}$,
  which is a contradiction since $\ell_i$ and $\ell_{i'}$ intersect $\ell$ strictly between $A$ and $B$,
  and cannot have two intersection points with $\ell$.
  If $j>1$, we have that $p_{j-1}p_j$ is a segment of $\ell_i$ or $\ell_{i'}$,
  which implies that $p_{j-1}$ is either in $\ell_{i}$ which contradicts $f(i)=j$,
  or in $\ell_{i'}$ which contradicts $f(i')=j$.
  %Since $f$ injective, we have $m\geq k$. So $p$ has (graph) length higher than $\ell_{AB}$. And a shortest graph path from $A$ to $B$ is indeed along the line $\ell$.
\end{proof}

When two vertices do not belong to a common grid line,
a shortest path can be found along two grid lines
(following one, then the other, in order to reach the second vertex from the first vertex).
The two grid lines are consecutive or adjacent,
which we define below, before stating the lemma.

\begin{definition}[Adjacent directions]
  Directions $\zeta_i$ and $\zeta_j$ are \emph{adjacent}
  if there is no other multigrid direction $\zeta_k$ in the cone delimited by $\zeta_i$ and $\zeta_j$.
  Formally, consider $(\epsilon_i \zeta_i)_{0\leq i <d}$ with $\epsilon_i = \pm 1$ such that $\mathrm{arg}(\epsilon_i\zeta_i)\in [0,\pi[$.
    Now order the $\epsilon_i\zeta_i$ by increasing argument in $[0,\pi[$.
      We say that $\zeta_i$ and $\zeta_j$ are adjacent directions if $\epsilon_i\zeta_i$ and $\epsilon_j \zeta_j$ are consecutive in this ordering,
      considering that the last and first ones are also consecutive.
  %%\FIXME{trouver une formulation à base d'ordre cyclique des directions dans le demi-cercle ? (all vectors have norm $1$)}.
\end{definition}

\begin{figure}
  \center 
  \includegraphics[width=0.2\textwidth]{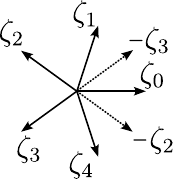}
  \caption{Adjacency of directions for the $5$-fold multigrid : $\zeta_0$ is adjacent to $\zeta_2$ and $\zeta_3$.}
  \label{fig:adjacent_directions}
\end{figure}

\begin{lemma}[Regular multigrid distance 2]
  Let $\zeta_i$ and $\zeta_j$ be adjacent multigrid directions,
  and let $\ell$ be a $i$-line and $\ell'$ be a $j$-line.
  Denote $C$ the intersection point of $\ell$ and $\ell'$.
  For any vertices $A$ of $\ell$ and $B$ of $\ell'$ such that $(C-A)\cdot(B-C)\geq 0$,
  the path from $A$ to $B$ in the multigrid that goes along $\ell$ until $C$ then along $\ell'$
  is a shortest path for the graph distance.
  %% We have :
  %% \begin{enumerate}
  %% \item the shortest path from $A$ to $B$ in the multigrid is along $l$ until $C$ then along $l'$,
  %% \item the distance from $A$ to $B$ is exactly the number of edges along $l$ between $A$ and $C$ plus the number of edges along $l'$ between $C$ and $B$.
  %% \end{enumerate}
  \label{lemma:multigrid_distance_2}
\end{lemma}

\begin{figure}
  \center
  \includegraphics[width=0.5\textwidth]{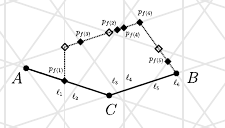}
  \caption{
    Three vertices $A$, $B$ and $C$, such that $C$ is at the intersection
    of the two adjacent multigrid directions.
    If the multigrid is regular then a shortest path between $A$ and $B$
    is along the two segments $AC$ then $CB$.
  }
  \label{fig:triangular_inequality_2}
\end{figure}

\begin{proof}
  See Figure~\ref{fig:triangular_inequality_2} for an illustration.
  Denote $\ell_{AC}$ and $\ell_{CB}$ the straight paths from $A$ to $C$ and from $C$ to $B$ respectively,
  and $\ell_{ACB}$ the union of $\ell_{AC}$ and $\ell_{CB}$.
  Let $(\ell_i)_{0<i\leq k}$ be the lines that intersect $\ell{AC}$ strictly between $A$ and $C$,
  or $\ell{CB}$ strictly between $C$ and $B$. %the interior of $\ell_{AC}$ or $\ell_{BC}$
  Note that since $\ell$ and $\ell'$ are adjacent directions,
  there is no line that crosses both $\ell$ between $A$ and $C$ and $\ell'$ between $C$ and $B$.
  
  Let $p= (A,p_1,\dots,p_m,B)$ be any path from $A$ to $B$,
  and denote $k$ the number of intermediate vertices between in $\ell_{ACB}$, excluding $A$, $B$ and $C$.
  Applying twice the same reasoning as in the proof of Lemma~\ref{lemma:multigrid_distance_1} gives $m\geq k$.
  Indeed, each line that crosses $\ell_{AC}$ or $\ell_{CB}$ also crosses the segment $AB$,
  because $\zeta_i$ and $\zeta_j$ are adjacent directions, and therefore it also crosses the path $p$.
  However the length of $\ell_{ACB}$ is $k+1$, as $C$ is also a vertex of $\ell_{ACB}$.

  We define $\ell_{k+1}$ as follows.
  Vertex $B$ is a crossing of $\ell'$ with some other line $\ell''$.
  As $p$ ends at $B$, the last edge of $p$ is either along $\ell$ or $\ell'$.
  If the last segment of $p$ is along line $\ell'$,
  then set $\ell_{k+1}:= \ell'$, otherwise set $\ell_{k+1} := \ell''$.
  With this extra convention we obtain an injective mapping $f$
  from $\llbracket 1, k+1\rrbracket$ to $\llbracket 1, m\rrbracket$.
\end{proof}

With this understanding of the distance between pairs of crossings
(vertices of the multgrid graph), %corresponding to tiles of the dual tiling
we are able to approximate the corona limit of a patch.
For this purpose we start again by first considering only the crossings on dominant lines,
and then include other crossings by exploiting our previous result (Lemma~\ref{lemma:multigrid_distance_2})
that distances can be approximated using two adjacent multigrid directions.
Let us introduce the notions of dominant line vertices, two lines vertices,
and their respective approximations of a given patch (their intersection with it).

\begin{definition}[Dominant line vertices and approximant]
  Given the dominant lines \\ %% NOTE retour à la ligne pour éviter un dépassement de marge
  $\dominantlines=(L_i)_{0\leq i < d}$,
  denote $\multigrid^L$ the set of vertices that are on a dominant line.
  Given a patch $D$, we call \emph{dominant line approximant} of $D$ the set of vertices
  $D\cap \multigrid^L$, \emph{i.e.}, the vertices in $D$ that are on a dominant line.
\end{definition}

An example of vertex set $\multigrid^L$ is presented on Figure~\ref{fig:dominant_lines}.

\begin{figure}
  \center
  \includegraphics[width=0.8\textwidth]{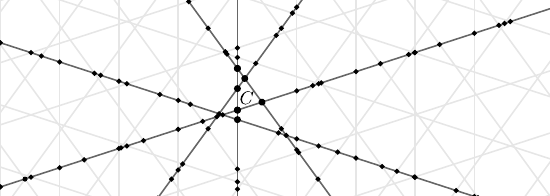}
  \caption{
    In round dots a patch $C$, with dominant lines shown in dark grey.
    The dominant line vertices $\multigrid^L$ are depicted in diamonds.
  }
  \label{fig:dominant_lines}
\end{figure}

\begin{lemma}[Dominant line approximant of the coronas]
  Let $C$ be a finite initial patch.
  There exists $\delta_2\in \mathbb{R}^+$ such that for any $n\in\N$,
  \[
    (n- \delta_2)\polygon \cap \multigrid^L
    \subseteq
    C_n \cap \multigrid^L
    \subseteq
    (n+\delta_2)\polygon\cap \multigrid^L.
  \]
  \label{lemma:dominant_line_approximant}
\end{lemma}

\begin{proof}
  Let $\dominantlines=(L_i)_{0\leq i<d}$ be the dominant lines of patch $C$.
  %% XXX : remark, if 0 not in the hull of C, it is only true for large enough $n$, what should we say?
  Denote $\diam(C)$ the diameter of $C$ for the graph distance.
  Define $\delta_2 = \delta_1 + \diam(C)$.
  Let us consider an arbitrary $n\in \mathbb{N}$.
  We first argue that
  \[
    (n-\delta_2)\polygon \cap \multigrid^L
    \subseteq
    (n-\delta_1)\polygon \cap \multigrid^L
    \subseteq
    C_n \cap \multigrid^L.
  \]

  Recall that by convention, if $(n-\delta_2)\leq 0$
  we have $(n-\delta_2)\polygon = \emptyset$ so the inclusion vacuously holds.
  %% Note de kevin (02/01/2024) : was $= \{0\}$.
  %% trivially since $\{0\}$ is in the hull of $C$ by hypothesis.
  In the following we assume $n> \delta_2$.
  On line $L_i$, the vertices in $(n-\delta_1)\polygon$ are also in $\overbow{E_n}$
  by Proposition~\ref{prop:endpoints_limit_shape}, so they are at graph distance at most $n$
  from the initial endpoint $e_{0,i}^+$ or $e_{0,i}^-$ on line $L_i$,
  which means that they are at graph distance at most $n$ from $C$.
  Consequently they are in $C_n$, so that
  \[
     (n-\delta_2)\polygon \cap \multigrid^L
    \subseteq
     C_n \cap \multigrid^L.
  \]

  We now prove the second inclusion:
  $C_n \cap \multigrid^L \subseteq (n+\delta_2)\polygon \cap \multigrid^L$.
  Consider a line $L_i$ and a vertex $z$ in $C_n \cap L_i$.
  %% \KP{[such a $z$ may not exist for small $n$]}
  If $z\in C$ we have in particular that $z\in \overbow{E_0}\subseteq \overbow{E_n}$,
  so by Proposition\ref{prop:endpoints_limit_shape} we have
  \[
    z\in  (n+\delta_1)\polygon \cap \multigrid^L
    \subseteq
    ( n+\delta_2)\polygon \cap \multigrid^L.
  \]
  Otherwise, if $z$ is not in the original patch $C$, $z$ is on a dominant some dominant line $L_i$ either in direction $+\zeta_i^\bot$ or in direction $-\zeta_i^\bot$.
  Without loss of generality, assume $z$ is in direction $+\zeta_i^\bot$ from $e_{0,i}^+$.
  %% QUESTION  \KP{[comprends pas la phrase... assume \emph{without loss of generality} ?]}.
  Since $z\in C_n$ there exists a vertex $z'\in C$ such that the graph distance $\distance(z,z')\leq n$.
  Since $C$ has diameter $\diam(C)$ we have $\distance(e_{0,i}^+, z')\leq \diam(C)$,
  so $\distance(e_{0,i}^+, z')\leq n+\diam(C)$.
  We then get that $z$ is between $e_{0,n}^+$ and $e_{n+\diam(C),i}^+$
  so that $z \in L_i \cap \overbow{E}_{n+\diam(C)}$,
  and by Proposition~\ref{prop:endpoints_limit_shape} we have
  $z \in l_i \cap (n + \diam(C) + \delta_1)\polygon$.
  Overall we get as expected $C_n  \cap \multigrid^L\subseteq (n+\delta_2)\polygon \cap \multigrid^L$
  with $\delta_2 = \delta_1 + \diam(C)$.
\end{proof}

\begin{definition}[Two lines vertices and approximant]
  Given a patch $C$ with dominant lines $\dominantlines=(L_i)_{0\leq i<d}$,
  and two adjacent grid directions $\zeta_i$ and $\zeta_j$,
  we denote by $\multigrid^{\diamond(i,j)}$ the set of multigrid vertices that are both
  in the cone delimited by the dominant lines $L_i$ and $L_j$, and
  either along a $i$-line or a $j$-line.
  We denote by $\multigrid^{\diamond}$ the union of $\multigrid^{\diamond(i,j)}$
  for all pairs of adjacent directions $i,j$.
  Given a patch $D$, we call \emph{two lines approximant} of $D$ the set of vertices
  $D\cap \multigrid^\diamond$.
\end{definition}

Examples of vertex sets $\multigrid^{\diamond(i,j)}$
and $\multigrid^{\diamond}$ are presented on Figure~\ref{fig:two_lines_vertices}.
%% Note that the only reference to patch $C$ in these definitions is through its dominant lines.
%% An alternative definition of the two lines approximant
%% \KP{could use an arbitrary line in each direction as reference.
%% [I think that some precisions are missing, even for the dominant lines 
%% according to an initial patch $C$, on the case when $C$ does not interset some gridline directions
%% (shoud we state the results ``for $n$ big enough'', or expose the matter with more clarity in Remarks 5 and 6 ?)]}
Lemma~\ref{lemma:two_lines_approximant} is the last difficult technical ingredient
towards the proof of our main result.

\begin{figure}
  \center
  \begin{subfigure}{.8\textwidth}
    \center
    \includegraphics[width=0.8\textwidth]{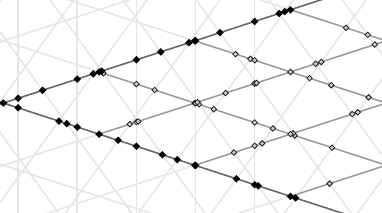}
    \caption{
      The two lines vertices $\multigrid^{\diamond(i,j)}$ in black and gery diamonds.
      The lines $L_i$ and $L_j$ on the boundary of the cone are in dark grey,
      and the lines parallel to $L_i$ or $L_j$ in the cone are in medium grey.
      Vertices outside of the cone or in the cone but not on a $i$-line or a $j$-line
      do not belong to $\multigrid^{\diamond(i,j)}$.
    }
  \end{subfigure}
  \\[2em]
  \begin{subfigure}{.8\textwidth}
    \center
    \includegraphics[width=0.9\textwidth]{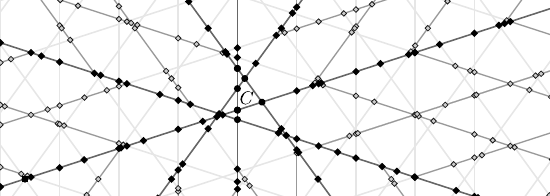}
    \caption{
      The two lines vertices $\multigrid^\diamond$ in black and grey diamonds.
      The dominant lines are shown in dark grey.
    }
  \end{subfigure}
  \caption{
    Illustration of two lines vertices
    $\multigrid^{\diamond(i,j)}$ and $\multigrid^{\diamond}$.
  }
  \label{fig:two_lines_vertices}
\end{figure}

\begin{lemma}[Two lines approximant of the corona]
  Let $C$ be a finite initial patch.
  %of dominant lines $(L_i)_{0\leq i <d}$, and $C_n$ its successive coronas.
  There exists $\delta_3$ such that for any $n\in\N$,
  \[
    (n-\delta_3)\polygon \cap \multigrid^\diamond
    \subseteq
    C_n\cap \multigrid^\diamond
    \subseteq
    (n+\delta_3)\polygon \cap \multigrid^\diamond.
  \]
  \label{lemma:two_lines_approximant}
\end{lemma}

\begin{proof}
  %% TODO corriger/reformuler preuve
  Let $\dominantlines=(L_i)_{0\leq i<d}$ be the dominant lines of patch $C$.
  Denote $\cone_{i,j}^+$ and $\cone_{i,j}^-$ the two half-cones %% NOTE they are convex cones, and not sets of vertices 
  delimited by lines $L_i$ and $L_j$ where $i$ and $j$ are adjacent directions,
  and $x_{i,j}$ the intersection of lines $L_i$ and $L_j$.

  %% QUESTION
  %% \begin{figure}
  %%   \center
  %%   \TODO{}
  %%   \caption{The characteristic polygon and its intersection with a cone.}
  %%   \label{fig:polycone}
  %% \end{figure}

  Remark that as the intersection of $x_{i,j}+\polygon$ with $\cone_{i,j}$ is a triangle, it is convex and therefore  
  in $\cone_{i,j}$, the $x_{i,j} + \alpha\polygon$ are in some sense linear.
  That is : if $x \in \cone_{i,j}\cap (x_{i,j} + \alpha\polygon)$ and $y \in (\cone_{i,j}-x_{i,j}) \cap \beta \polygon$
  then $x+y\in \cone_{i,j}\cap \left(x_{i,j} + (\alpha + \beta)\polygon\right)$.
  Recyprocally, if $x \in \cone_{i,j} \cap (x_{i,j}+ \alpha\polygon)$, and $x+y \in \cone_{i,j} \cap \left(x_{i,j} + (\alpha+\beta)\polygon\right)$ with $y\in \cone_{i,j}-x_{i,j}$ then $y \in \beta\polygon$.

  \medskip

  Denote $\dominantlineoffset$ the ratio between the maximum of the modulus of the intersection points of the dominant lines of adjacent directions and the minimal modulus of the vertices $\polygon_i$ of the characteristic polygon $\polygon_i$, \emph{i.e.},
  \[
    \dominantlineoffset := \max\limits_{i,j \text{ adjacent directions}} |x_{i,j}| / \min\limits_{i} |\polygon_i|. %% FIXME
  \]
  In particular, we have that all $x_{i,j}$ are in $\dominantlineoffset\polygon$, which we will use later.  
  Define
  \[
    \delta_3 = \lceil\delta_2\rceil + \lceil \delta_1\rceil + \lceil \dominantlineoffset \rceil + \dominantlineoffset +2.
  \]
  
  \smallskip
  We prove the first inclusion.
  As $\multigrid^\diamond$ is simply the union of the $2d$ half-cones $\cone_{i,j}^\pm$, we prove it separately for each half cone.
  Let $i$ and $j$ be adjacent directions, let $x\in \multigrid^\diamond$ in the half cone $\cone_{i,j}^+$ such that $x\in (n -\dominantlineoffset) \polygon$.
  By definition of $\dominantlineoffset$ we have $(n-\dominantlineoffset)\polygon \subset x_{i,j} + n\polygon$
  and therefore $x\in x_{i,j} + n\polygon$.
  By definiton of $\multigrid^\diamond$, $x$ is the intersection of a $i$-line or a $j$-line in a half-cone $\cone_{i,j}^+$ or $\cone_{i,j}^-$.
  Without loss of generality, assume $x$ on a $i$-line which we denote $\ell_i$.
  If $\ell_i$ is actually the dominant line $L_i$,
  the inequality holds due to Lemma~\ref{lemma:dominant_line_approximant}.
  Otherwise, the line $\ell_i$ crosses the other dominant line $L_j$ at some point $x'$.
  Because $i$ and $j$ are adjacent directions,
  we have $x' \in (x_{i,j} + \alpha n \polygon)$ and $x \in \ell_i \cap (x' + (1-\alpha)n \polygon$ for some $\alpha \in [0,1]$.
  It follows by Lemma~\ref{lemma:dominant_line_approximant} that $x' \in C_{\lceil \alpha n \rceil + \lceil\delta_2\rceil \lceil \dominantlineoffset \rceil}$
  and from Corollary~\ref{corollary:nth_vertex} that  $\distance(x,x') \leq \lceil (1-\alpha) n \rceil  \lceil \delta_1 \rceil$.
  Overall, with Lemma~\ref{lemma:multigrid_distance_2} we obtain that $x \in C_{n+\lceil \delta_2\rceil  + \lceil\delta_1\rceil + \lceil \dominantlineoffset \rceil + 2}$.
  
  We have proved that for any $n$, $(n-\dominantlineoffset)\polygon \cap \multigrid^\diamond \subset  C_{n+\lceil \delta_2\rceil + \lceil\delta_1\rceil  \lceil \dominantlineoffset \rceil + 2}$.
  We can reformulate it as for any $n$, $(n-\delta_3)\polygon \cap \multigrid^\diamond \subset  C_{n}$,
  as expected.
  
  \medskip 
  Now we prove the second inequality.
  Take $x \in C_{n}\cap \multigrid^\diamond$.
  If $x$ is on a dominant line, the inequality holds due to Lemma \ref{lemma:dominant_line_approximant}.
  We now assume $x$ is not on a dominant line.
  By definition there exist adjacent directions $i,j$ such that $x$ is in $\cone_{i,j}^\pm$.
  Without loss of generality, assume that $x$ is on some $j$-line $\ell_j$ in $\cone_{i,j}^+$.
  Denote $x' = L_i \cap \ell_j$.
  Recall that, without loss, we assume that the original patch $C_0$ intersects all dominant line.
  Let $x''$ be a point of $C_0\cap L_i$.
  By Lemma \ref{lemma:multigrid_distance_2},
  a shortest path from $x''$ to $x$ is through $x'$.
  Consequently, $\distance(x'',x) = \distance(x'',x') + \distance(x',x)$ and since $x\in C_n$ we have $\distance(x'',x)\leq n$.
  So there exists $k<n$ such that $\distance(x'',x')=k$ and $\distance(x',x)\leq n-k$.
  Applying Lemma~\ref{lemma:dominant_line_approximant} we have
  \[
    x' \in C_k\cap \multigrid^L
    \subseteq
    (k+\delta_2)\polygon
  \]
  and by Corollary~\ref{corollary:nth_vertex} we have
  $x-x' = a \zeta_j^\bot$ with $a \leq (n-k+\delta_1)\polygon_j$,
  which means that $x \in x' + (n-k+\delta_1)\polygon$.
  Since $x$ is in the cone $\cone_{i,j}^+$ and $x-x'$ is in the $0$-centered cone $\cone_{i,j}^+-x_{i,j}$,
  we obtain
  \[
    x \in (k+\delta_2 + (n-k)+\delta_1)\polygon \subseteq (n + \delta_3)\polygon.
  \]
\end{proof}

The following simple result shows that crossings of a given type $(i,j)$ are relatively dense,
in the sense that there is a fixed radius (in terms of Euclidean distance or graph distance)
such that around any crossing involving an $i$-line, one can find a crossing of type $(i,j)$.
From this relative density, we obtain the relative density of $\multigrid^\diamond$,
so that the two lines approximants of the coronas are close to the actual coronas
(as formalized in Proposition~\ref{prop:multigrid_corona_limit}).
%% \KP{Why is it useful ?
%% We first did dominant lines: quite sparse.
%% Then we considered two lines vertices: they are dense,
%% and indeed characterize finely the corona limit.}

\begin{lemma}[Multigrid bounded distance]
  In the multigrid $\multigrid$,
  there exists a radius  $r'\in \mathbb{R}^+$ and $k_r'\in\mathbb{N}$ such that
  for any two distinct grid directions $\zeta_i$ and $\zeta_j$,
  and any multigrid vertex $z$ on a $i$-line,
  there is a vertex $z'$ intersection of that same $i$-line with a $j$-line,
  such that the Euclidean distance from $z$ to $z'$ is at most $r'$
  and the graph distance from $z$ to $z'$ is at most $k_r'$.
  \label{lemma:density}
\end{lemma}

\begin{proof}
  Simply take $r' := \max_{i\neq j} (\zeta_i^\bot \cdot \zeta_j)^{-1}$
  the maximum over $j\neq i$ of the distance along some $i$-line between two consecutive crossings of type $(i,j)$,
  and $k_r' := (d-1) \lceil r\rceil$ because between two vertices at Euclidean distance $r'$ on some $i$-line,
  there are at most $\lceil r'\rceil$ crossings of type $(i,j)$ for each $j\neq i$.
\end{proof}

\begin{corollary}
  In the multigrid $\multigrid$,
  there exist a radius $r\in \mathbb{R}^+$ and $k_r\in\mathbb{N}$ such that
  for any two distinct grid directions $\zeta_i$ and $\zeta_j$,
  and any multigrid vertex $z$,
  there is a vertex $z'$ intersection $(i,j)$,
  such that the Euclidean distance from $z$ to $z'$ is at most $r$
  and the graph distance from $z$ to $z'$ is at most $k_r$.
  \label{corollary:density}
\end{corollary}
\begin{proof}
  Here we take $r:=2r'$ and $k_r := 2k_r'$ with $r'$ and $k_r'$ from Lemma \ref{lemma:density}.
  Let $z$ be a multigrid vertex.
  $z$ is on some $k$-line $\ell$. At euclidean distance at most $r'$ and graph distance at most $k_r'$ from $z$ there is a point $z''$ of type $(i,k)$, and at euclidean distance at most $r'$ and graph distance at most $k_r'$ from $z''$ there is a point $z'$ of type $(i,j)$.
  Overall $z'$ is at euclidean distance at most $r=2r'$ at graph distance at most $k_r=2k_r'$ from $z$.
\end{proof}

The next Proposition is the main result of this section,
concluding our study of the corona limit in the multigrid.
It states that from any initial patch in the multigrid,
successive coronas tend to what we have defined as the characteristic polygon of the multigrid.
Recall that $\multigrid$ is a regular multigrid, seen as a graph (its set of vertices).

\begin{proposition}[Multigrid corona limit]
  The multigrid corona limit is the characteristic polygon $\polygon$.
  For any multigrid patch $C$, there exists a constant $\delta_4\in \mathbb{R}^+$ such that 
  \[
    (n-\delta_4)\polygon \cap \multigrid
    \subseteq
    C_n
    \subseteq
    (n+\delta_4)\polygon \cap \multigrid.
  \]
  \label{prop:multigrid_corona_limit}
\end{proposition}

\begin{proof}
  This proposition results from combining Lemma~\ref{lemma:two_lines_approximant}
  on the shape of the two lines approximant and Corollary~\ref{corollary:density}
  on the relative density of vertices of a given type.
  For the sake of completeness we give the full details for the first inclusion.

  Define $\delta_4 = \delta_3 + r + k_r$.
  Take $z \in (n-\delta_4)\polygon \cap \multigrid$.
  Vertex $z$ is in some $\cone_{i,j}$ for some adjacent directions $i$ and $j$.
  By lemma~\ref{lemma:density}, there is a vertex $z'$ of type $i$ or $j$
  at Euclidean distance at most $r$ and graph distance at most $k_r$ from $z$.
  Note that since the boundaries of $\cone_{i,j}$ are a $i$-line and a $j$-line,
  we can also take $z'$ in $\cone_{i,j}$.
  We then have
  \[
    z'\in (n-\delta_4+r)\polygon \cap \multigrid^\diamond
    \quad\text{\emph{i.e.},}\quad
    z'\in (n-k_r-\delta_3)\polygon \cap \multigrid^\diamond.
  \]
  By Lemma \ref{lemma:two_lines_approximant} we have $z'$ in $C_{n-k_r}$, and therefore $z\in C_n$.
  The second inclusion follows by a similar argument.
\end{proof}

%%%%%%%%%%%%%%%%%%%%%%%%%%%%%%%%
\subsection{Corona in the tiling}
\label{subsec:corona-tiling}

In order to transfer Proposition~\ref{prop:multigrid_corona_limit} on the multigrid graph $\multigrid$,
to the edge-to-edge rhombus tiling $\tiling$,
we employ the dualization function $\dualization$.
However, as an intermediate step, we define a linear function which is almost $\dualization$.

\begin{definition}[Linear almost dual function]
  We define the function $\lindual$ as
  \[
    \lindual(z) := \sum\limits_{i\leq 0<d} (z\cdot\zeta_i) \zeta_i.
  \]
\end{definition}

\begin{lemma}[The dualization is almost linear]
  The dualization function $\dualization$ is at uniformly bounded distance from the linear function $\lindual$.
  More precisely, for any $z \in \CC$, $|\dualization(z) - \lindual(z)| \leq 2d$.
  \label{lemma:almost_linear}
\end{lemma}

\begin{proof}
  Recall that for any $i$, $|\zeta_i|=1$ and $0\leq \gamma_i < 1$. Take $z \in \CC$.
  We have
  \[
    \dualization(z)-\lindual(z) =
    \sum\limits_{i\leq 0<d} \left( \left\lfloor z\cdot\zeta_i-\gamma_i \right\rfloor - z\cdot \zeta_i \right)\zeta_i
  \]
  and
  \[
    \left| \left\lfloor z\cdot \zeta_i-\gamma_i \right\rfloor - z\cdot \zeta_i \right| \leq 2.
  \]
  Consequently,
  \[
  |\dualization(z)-\lindual(z)|
  \leq
  \sum\limits_{i\leq 0<d} \left| \left( \left\lfloor z\cdot\zeta_i-\gamma_i \right\rfloor - z\cdot\zeta_i \right)\zeta_i\right|
  \leq
  \sum\limits_{i\leq 0<d} 2
  \leq
  2d
  .\]
\end{proof}

The point of Lemma~\ref{lemma:almost_linear} is that,
if $\Delta$ is the corona limit of a multigrid,
then its linear almost-dual $\lindual(\Delta)$ is the corona limit in the dual tiling,
from which we obtain the main theorem.

\setcounter{theorem}{0}
\begin{theorem}[Corona limit of a regular multigrid dual tiling]
  The corona limit of a regular multigrid dual tiling is its characteristic polygon $\dual{\chi}$.
\end{theorem}

\begin{proof}
  Let $\multigrid(\zeta,\gamma)$ be a regular multigrid and $\tiling(\zeta,\gamma)$ its dual tiling.
  Let $P_0$ be a finite patch in $\tiling(\zeta,\gamma)$.
  Denote $C_0$ the set of multigrid vertices in $\multigrid(\zeta,\gamma)$ such that $\dual{C_0} = P_0$.
  By Proposition~\ref{prop:multigrid_corona_limit}, the corona limit of $C_0$ is $\chi$.
  We now prove that the corona limit of $P_0$ is $\lindual(\polygon)=\dual{\polygon}$.

  Take $n\in\N$, denote $P_n$ and $C_n$ the $n$-th corona of $P_0$ and $C_0$.
  We have that $\lindual(\overbow{C_n})$ is within bounded distance of $P_n$.
  Indeed,, each multigrid cell in $\overbow{C_n}$ is the dual of a vertex of a tile in $P_n$,
  and each multigrid cell that is edge adjacent to $\overbow{C_n}$ is the dual of a boundary vertex of $P_n$,
  and each multigrid cell that is neither in or edge-adjacent to $\overbow{C_n}$ is not a vertex in $P_n$.
  Recall that if $c$ is a cell in $\multigrid(\zeta,\gamma)$
  and $x$ is its dual vertex in $\tiling(\zeta,\gamma)$,
  then for all $z\in c$ we have $\dualization(z)=x$.
  Now, by applying Lemma~\ref{lemma:almost_linear},
  we obtain that the Hausdorff distance between $\overbow{P_n}$ and $\lindual(\overbow{C_n})$ is at most $2d$.
  From this we have $\hausdorff\left(\overbow{P_n}/n , \lindual(\overbow{C_n})/n\right) \leq 2d/n$.
  Since $\lindual$ is linear we have $\lindual(\overbow{C_n})/n = \lindual(\overbow{C_n}/n)$.
  Recall also that by Proposition \ref{prop:multigrid_corona_limit} we have $\overbow{C_n}/n \tendsto{n} \polygon$.
  By continuity of $\lindual$ we have $\lindual(\overbow{C_n}/n) \tendsto{n} \lindual(\polygon)$.
  Overall we obtain $\overbow{P_n}/n \tendsto{n} \lindual(\polygon)$, that is : the corona limit of $P_0$ is $\lindual(\chi)$.
  We have proved that the corona limit of any finite patch is $\lindual(\polygon)=\dual{\polygon}$.
\end{proof}

%%%%%%%%%%%%%%%%%%%%%%%%%%%%%%%%
\section{Conclusions and perspectives}
\label{s:conclusion}

In this article, we have greatly generalized the fact that, from any initial patch of tiles
in a Penrose tiling, the edge-to-edge growth process always tends to a regular decagon~\cite{ai16}.
Indeed, for any multigrid dual tiling, Theorem~\ref{theorem:corona} proves
that this corona limit tends to a polygon with parallel opposite sides,
whose orientations are given by the multigrid normal vectors (it is independent of the offsets).
We call it the characteristic polygon, and formulate its coordinates in the multigrid
and in the dual tiling (Definition~\ref{def:polygon}).
To explain this phenomenon, we have explored the corresponding growth process in the multigrid,
where the successive coronas can be ``averaged'' thanks to the regular spacing of multigrid lines
(straightforward on individual lines, but not trivial to combine).
Our characterization of the corona limit smoothly transfers from the multigrid to its dual tiling,
because the dualization is almost linear (Lemma~\ref{lemma:almost_linear}).

Besides edge-to-edge propagation, it may be meaningful in certain contexts to study
the growth process relative to a given neighborhood, in the spirit of~\cite{gg93}.
The notion of neighborhood makes sense in multigrid dual tilings,
because they have finite local complexity (for the atlas of Penrose tilings, see~\cite{s96,fl23}).
The neighborhood relation may also be ``directionnal'', meaning that the orientation of rhombuses
may be taken into account, for example by considering only the propagation in the north direction.
Early simulations on Penrose tilings with fat and thin rhombuses
suggest that the corona limit is still the regular decagon
$\dual{\polygon}$ (see Figure~\ref{fig:corona_example_big})
when the neighborhood is:
\begin{itemize}
  \item restricted to different tile types only (thin with fat, and fat with thin, which is symmetric),
  \item removed for fat with fat tile types (which is symmetric),
  \item removed for thin with thin tile types (which is symmetric),
  \item restricted to the postive direction
    (which is not symmetric: out-neighbors are taken only on one side of the tile,
    according to the orientation of the multigrid lines at the corresponding crossing).
\end{itemize}
When the neighborhood is restricted to one cardinal direction
(\emph{i.e.} directed along a half-plane, which is not symmetric),
we observe the polygon $\dual{\polygon}$ truncated
(no propagation occurs in the direction opposite to the half-plane).

Another track of research would be to pursue the generalization of $\Z^2$ dynamics
to multigrid dual tilings (in particular the most well known Penrose tilings,
but arguments may not be particular to one set of normal vectors and offsets)
with percolation processes.
Successive coronas correspond to the propagation process where a tile is included in the selection
(often refered to as being ``infected'') whenever it has at least one of its four adjacent tiles already selected.
Bootstrap percolation is the propagation process where a tile is infected whenever it has at least two
adjacent tiles (or more generaly, neighbors) already infected.
From an initial configuration $c\in \{0,1\}^{\mathbb{Z}^2}$ taken at random along a Bernoulli distribution of parameter $p>0$,
the probability of percolation (contaminating the full $\mathbb{Z}^2$) is 1 \cite{ve87}.
The same behaviour of almost sure percolation from any initial configuration taken at random along a Bernoulli of positiev parameter $p$ appears to generalize to Penrose and multigrid dual tilings.

On $\mathbb{Z}^2$, percolation processes follow a $0-1$ law \cite{ve87},
that is, for any percolation process the probability of percolation is either $0$ or $1$ (and never in between)
when the initial configuration is taken along a Bernoulli distribution.
We conjecture that on sufficiently regular rhombus tilings,
including Penrose and multigrid dual tilings, a similar $0-1$ law holds.

%%%%%%%%%%%%%%%%%%%%%%%%%%%%%%%%
\section*{Acknowledgments}

This work received support from ANR-18-CE40-0002 FANs projet,
STIC AmSud CAMA 22-STIC-02 (Campus France MEAE) project,
HORIZON-MSCA-2022-SE-01-101131549 ACANCOS project,
ANR JCJC 2019 19-CE48-0007-01 C\_SyDySi project,
RIN DynNet project, %% TODO find exact code number
and a postdoctoral grant from Institut Archimède (Aix Marseille University).

%%%%%%%%%%%%%%%%%%%%%%%%%%%%%%%%
\bibliographystyle{alpha}
\bibliography{biblio}

%%%%%%%%%%%%%%%%%%%%%%%%%%%%%%%%
\end{document}